\pdfoutput=1
\documentclass[
superscriptaddress,
nofootinbib,longbibliography,
onecolumn,
notitlepage,
aps]
{revtex4-2}

\usepackage{amsmath,amsfonts,amssymb,amsthm,amsbsy,mathtools}

\usepackage{graphicx}
\usepackage{dcolumn}
\usepackage{bm}
\usepackage{physics} 
\usepackage{bbold}
\usepackage{mathtools}
\usepackage{xcolor}
\usepackage{tikz}
\usepackage{tikz-cd}
\usepackage{comment}
\usepackage{subfigure}
\usepackage{comment}
\usepackage{soul}
\usepackage{framed}
\usepackage{mdframed}
\usepackage{csquotes}
\usepackage{appendix}
\usepackage{hyperref}

\usepackage{natbib}
\usepackage[normalem]{ulem}
\usepackage{xcolor}
\usepackage{xspace}
\usepackage{ifthen}
\usepackage{braket}
\usepackage{mathrsfs}
\usepackage{thmtools, thm-restate}

\DeclareMathAlphabet{\mathscrbf}{OMS}{mdugm}{b}{n}

\newcommand{\showcomments}{true}

\newcommand{\feb}[1]%
{\ifthenelse{\equal{\showcomments}{true}}
{{\color{teal}{#1}}}{\xspace}}%

\newcommand{\acd}[1]%
{\ifthenelse{\equal{\showcomments}{true}}%
{{\color{magenta}{#1}}}{\xspace}}%

\newcommand{\car}[1]%
{\ifthenelse{\equal{\showcomments}{true}}%
{{\color{blue}{#1}}}{\xspace}}%

\newcommand{\vk}[1]%
{\ifthenelse{\equal{\showcomments}{true}}%
{{\color{orange}{#1}}}{\xspace}}%

\newcommand{\jnb}[1]%
{\ifthenelse{\equal{\showcomments}{true}}%
{{\color{black}{#1}}}{\xspace}}%

\newcommand{\hg}[1]%
{\ifthenelse{\equal{\showcomments}{true}}%
{{\color{olive}{#1}}}{\xspace}}%

\theoremstyle{plain}
\newtheorem{proposition}{Proposition}

\newtheorem{corollary}{Corollary}

\newtheorem{definition}{Definition}

\makeatletter
\newcommand{\bigplus}{%
  \DOTSB\mathop{\mathpalette\mattos@bigplus\relax}\slimits@
}
\newcommand\mattos@bigplus[2]{%
  \vcenter{\hbox{%
    \sbox\z@{$#1\sum$}%
    \resizebox{!}{0.9\dimexpr\ht\z@+\dp\z@}{\raisebox{\depth}{$\m@th#1+$}}%
  }}%
  \vphantom{\sum}%
}
\makeatother

\begin{document}
\title{Causality is rare:\\ some topological properties of causal quantum channels}
\author{Robin Simmons}
\email{robin.simmons@univie.ac.at}
\affiliation{University of Vienna, Faculty of Physics, Vienna Doctoral School in Physics, and Vienna Center for Quantum Science and Technology (VCQ), Boltzmanngasse 5, A-1090 Vienna, Austria}
\affiliation{Institute for Quantum Optics and Quantum Information (IQOQI),
Austrian Academy of Sciences, Boltzmanngasse 3, A-1090 Vienna, Austria}
\begin{abstract}
    Sorkin's impossible operations demonstrate that causality of a quantum channel in QFT is an additional constraint on quantum operations above and beyond the locality of the channel. What has not been shown in the literature so far is how much of a constraint it is. Here we answer this question in perhaps the strongest possible terms: the set of (normal) causal channels is nowhere dense in the set of local (normal) channels. We connect this result to quantum information, showing that the set of causal unitaries has Haar measure $0$ in the set of all unitaries acting on a lattice. Finally, we close with a discussion of the implications and connections to recent QFT measurement models. 
\end{abstract}
\maketitle

\tableofcontents
\section{Introduction}
Sorkin's impossible operations \cite{sorkin1993impossible} point out a tension with the standard interpretation of algebras in quantum theory, and specifically algebraic quantum field theory (AQFT). Impossible operations are local---meaning acting non-trivially on only a bounded subset of any Cauchy hypersurface---quantum channels that allow superluminal signalling, and so are natural to exclude, leading to the idea that \textit{local} quantum channels need not be \textit{causal} quantum channels. Requiring that no physical channel can be acausal then suggests that the set of physical channels is a proper subset of all local channels. 

Building on this insight, several approaches have arisen. Borrowing techniques from quantum information, a detector-based approach has provided a number of concrete results regarding measurement channels and update results in QFT \cite{PhysRevD.105.065003,rel_causality_detec,particle_field_duality, PhysRevD.110.025013}. Blending algebraic QFT (AQFT) and quantum measurement theory, the Fewster-Verch (FV) framework has further formalised measurement theory in QFT \cite{fewster2020quantum}, and has been used to prove that every observable in a large class of QFTs can be measured in principle, and that measurement models based only on coupled QFTs are free from Sorkin's problems \cite{fewster_asymptotic_2023, PhysRevD.103.025017, mandrysch2024quantumfieldmeasurementsfewsterverch, fewster_measurement_2025}. Recently, several approaches have begun to investigate compositional properties of quantum evolutions in QFT, and their relation to causality \cite{oeckl2024spectraldecompositionfieldoperators, oeckl2025causalmeasurementquantumfield, simmons2025factorisationconditionscausalitylocal}. However, despite this flurry of research, very little is known about the set of causal quantum channels, outside of finite dimensional quantum mechanical analogies: see \cite{PhysRevA.64.052309, eggeling2002semicausal} for general statements about the relationship between causal and local operations, and \cite{PhysRevA.49.4331,PhysRevLett.90.010402, gisin2024towards} for proofs that all observables, even non-local, can be measured causally. In this paper, we aim to uncover some basic properties of the causal channels in QFT, although the methods here can be easily applied directly to other quantum models with much the same conclusions. The main informal question we wish to answer is
\begin{center}
    \textit{how rare is causality? }
\end{center}
Intuition from the quantum mechanical analogies suggests that it is quite rare, in the sense that picking a channel at random has a very small, in fact $0$ as we will show, probability of being causal with respect to a given tensor decomposition. However, formalising this statement in QFT requires significantly more effort. We will take \textit{nowhere dense} and \textit{meagre} as natural generalisations of the probabilistic notion of rare, where a nowhere dense set is a set where the interior of the closure is empty, and a meagre set can be written as a countable union of nowhere dense sets. Both of these notions can be thought of as expressing the rareness of a subset within a larger topological set. ´

In order to prove the main results of this work, we study the Banach space structure of (normal) completely bounded maps between von Neumann algebras, and their preduals, using operator space theory. This allows us to prove closure and compactness results for the convex set of (normal) unital completely positive maps, also known as quantum channels.

We pay special attention to normal channels, which are a natural choice from a quantum information standpoint. Normal channels have many of the properties that their finite dimensional counterparts have, such as a quantum Stinespring representation\footnote{All quantum channels with codomain $\mathcal{B}(\mathcal{K})$ have an algebraic Stinespring dilation, however the formulation in terms of an ancilla that is traced over requires normality.} and Kraus decomposition when their co-domain is $\mathcal{B}(\mathcal{K})$. Their motivation in the context of QFT is less obvious, in part due to the infinite number of non-unitarily equivalent representations of local $C^*$-algebras. In the sequel we will assume that the local algebras are von Neumann subalgebras of a fixed Hilbert space, a situation that is natural if there is a good choice of a vacuum state. From a topological perspective, the use of von Neumann algebras is significantly more convenient than $C^*$-algebras, which is why we mostly focus on them. We leave the relaxation of these choices to future work, and point the interested reader to \cite{redei_how_2010,redei_when_2010} for discussion on operational notions in QFT without the assumptions of von Neumann algebras and normal channels.

Using the topological results, we prove several ``rarity" properties of causal (normal) local channels within the set of all (normal) local channels in appropriate topologies. In the most general case, we prove that the causal (normal) local channels are nowhere dense in the set of all (normal) local channels. We also prove a probabilistic result that shows that picking a unitary acting on a lattice of finite dimensional systems (using the Haar measure) will be acausal with probability $1$. 
\begin{figure}[h]
    \centering
    \includegraphics[width=0.65\linewidth]{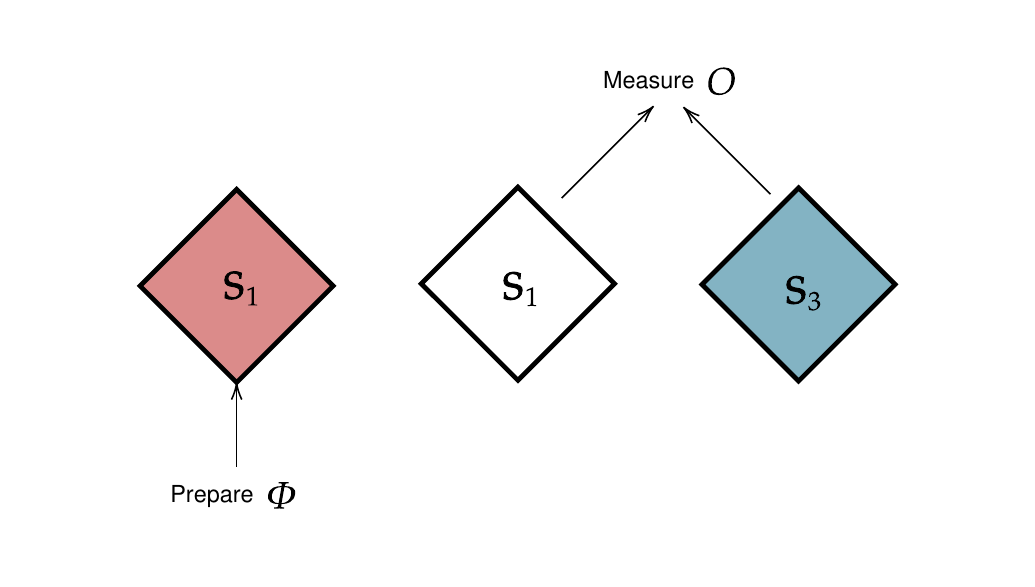}
    \caption{A spacetime scenario with $N=3$ spacelike separated systems. In this case, the first system locally prepares using $\Phi$, and the last two measure a joint observable $O$. We request that the effect of the preparation should be invisible to the expectation value of $U^\dagger OU$ if $U$ is causal.}
    \label{fig2}
\end{figure}
The paper is structured as follows. In section \ref{sec:motivation}, we outline the motivation from finite quantum information, including proving that the set of causal unitaries for a finite number of subsystems forms a measure $0$ subgroup of all unitaries. In section \ref{sec:results}, we state our main result, that normal causal channels are nowhere dense in normal local channels and sketch the proof. From there, the rest of the paper builds the background and propositions needed to prove the main result.

Sections \ref{sec:banach} to \ref{sec:closure_compact} introduce and develop the necessary Banach and operator space theory to approach theorem \ref{thrm:causal_meagre}, while section \ref{sec:qft} recalls the definition of an algebraic QFT (AQFT) and the formulation of Sorkin's impossible operations therein. A reader comfortable with Banach and operator space theory can skip to section \ref{sec:channels_topo}, where non-standard notations are introduced. Many of the results on the topological properties of quantum channels may also be of some general interest, as we are not aware of their statement elsewhere. 

\section{Motivation}\label{sec:motivation}
Finding causal channels by hand, for example in QFT, is highly \textit{non-trivial}, and most obvious choices are acausal, see \cite{PhysRevD.105.025003, albertini2023ideal, simmons_brukner}. For unitary channels on a finite collection of finite Hilbert spaces that are assumed to be spacelike separated, this rareness can be sharply stated, which we do now. Consider $N\geq 2$ spacelike separated systems, described by finite dimensional Hilbert spaces of dimension $\infty >d_1, d_2,\dots d_N>1$, see Fig \ref{fig2}, and so causality implies that \textit{no signals can be sent between the systems}. The full set of unitaries that act on the set of systems is $\text{U}(d_1d_2\dots d_N)$, while the set of unitaries that do not allow signals to be sent between the systems, and are thus causal, are those that can be written as 
\begin{equation}\label{eq:local_unit}
    U=U_1\otimes U_2\otimes\dots \otimes U_N,
\end{equation}
i.e. the local unitaries. By not being able to send a signal, we mean for any partition of the $N$ systems into $M$ and $N-M$ systems, any preparation channel $\Phi$ on the first partition and any operator $O$ on the second partition,
\begin{align}
    \Phi&=\Phi_M\otimes \text{id}_{N-M}, \;\; \Phi_i\in \text{UCP}\left(\bigotimes_{i=1}^M M_{d_i}\right), \\O&=\mathbb{I}_M\otimes O_{N-M}, \:\: O_{N-M}\in \bigotimes_{j=M+1}^{N}M_{d_j}(\mathbb{C}), 
\end{align}
we have
\begin{equation}\label{eq:qm_causal}
    \text{tr}(\rho \Phi(U^\dagger O U))-\text{tr}(\rho U^\dagger O U)=0,
\end{equation}
for all states $\rho$. Eq \eqref{eq:qm_causal} is a specific form of Sorkin's causality condition \cite{sorkin1993impossible}, also sometimes known as causal transparency. We will see the full QFT version, which takes the same form, in section \ref{sec:qft}.  Informally, we imagine $M$ systems conspiring beforehand to prepare a state with $\Phi$ in their shared past, and $N-M$ systems agreeing to meet in their shared future to measure $O$, again see Fig \ref{fig2}. Note we are working in Heisenberg picture, so the composition of maps is backwards with respect to time. Further, $\text{UCP}$ is the set of all unital completely positive maps, which for finite dimensions are simply the duals to the completely positive trace preserving channels on density operators. The requirement that the preparation $\Phi$ drops out of the expectation value of $U^\dagger OU$ leads us to conclude that $U$ factorises into a local unitary \cite{PhysRevA.64.052309}. 

Viewing the local unitaries as a subgroup of $\text{U}(d_1d_2\dots d_N)$, denoted $\text{U}(d_1, d_2, \dots, d_N)$, we must account for the phase ambiguity in the decomposition in Eq \eqref{eq:local_unit}, i.e. that $U_i\mapsto U_ie^{i\theta_i}, U\mapsto U$, if $\theta_i\in [0,2\pi)$ and $\sum_{i=1}^N \theta_i \text{ mod } 2\pi=0$. Hence, 
\begin{equation}
    \text{U}(d_1, d_2, \dots, d_N)\cong \text{U}(d_1)\times \text{U}(d_2)\times \dots \text{U}(d_N)/\text{U}(1)^{N-1}.
\end{equation}
We can now rephrase the question \textit{  how small is the set of causal unitaries\footnote{Since the causal unitary channels are defined by unitary operators only up to a phase, the set of unitary channels is in fact $\cong \text{U}(d_1)\times \text{U}(d_2)\times \dots \text{U}(d_N)/\text{U}(1)^{N}$, however, by the monotonicity of measures, this does not change our conclusions. } on $N$ finite dimensional spacelike separated subsystems?} as \textit{how small is $\text{U}(d_1, d_2, \dots, d_N)\cong \text{U}(d_1)$ in $\text{U}(d_1d_2\dots d_N)$?} We first answer this question for $N=2$, and bootstrap it to general $N$.
\begin{proposition}\label{prop:measure_causal}
    Let $n,m>1$ and define $\operatorname{U}(nm)\supseteq\operatorname{U}(n,m)\cong (\operatorname{U}(n)\times \operatorname{U}(m))/\operatorname{U}(1)$, then $\mu(\operatorname{U}(n,m))=0$ where $\mu$ is the Haar measure for $\operatorname{U}(nm)$.

    \begin{proof}
        Firstly, we note that $\text{U}(n,m)\subsetneq \text{U}(nm)$ is a strict closed subgroup of dim $n^2+m^2-1< n^2m^2$ when $n,m>1$. Since $\text{U}(nm)$ is connected, $\text{U}(n,m)$ cannot also be open. By the Steinhaus-Weil theorem, given a locally compact group $G$, a positive Haar measure $\mu_G(H)>0$ for a closed subgroup $H$ implies that the subgroup is also open. Since $\text{U}(n,m)$ is not open, $\mu(\text{U}(n,m))=0$.
    \end{proof}
\end{proposition}
This result leads us to conclude that the causal unitaries on two finite dimensional spacelike separated systems have measure $0$ in the set of all unitaries. We can apply proposition \ref{prop:measure_causal} to characterise the causal unitaries  for $N\geq 2$. Given $N$ sites each with local dimension $\infty>d_i>1$, the set of causal unitaries, i.e. unitaries that do not signal between the subsystems, is the closed strict subgroup $\text{U}(d_1,d_2,\dots,d_N)$, which obeys
\begin{equation}
    \text{U}(d_1d_2\dots d_N)\supsetneq\text{U}(d_1,d_2,\dots,d_N)\cong \text{U}(d_1)\times \dots \text{U}(d_N)/\text{U}(1)^{N-1}.
\end{equation}
Notably, we have the following strict inclusions of closed subgroups, 
\begin{equation}
    \text{U}(d_1d_2\dots d_N)\supsetneq \text{U}(d_1, d_2d_3\dots d_N)\supsetneq\text{U}(d_1,d_2,\dots, d_N)
\end{equation}
and so by proposition \ref{prop:measure_causal} and monotonicity of measures we have 
\begin{corollary}\label{corr:causal_unitaries_measure_0} For $N\in \mathbb{N}$, $1<d_1, d_2,\dots d_N<\infty$,
    \begin{equation}
    \mu(\text{U}(d_1,d_2,\dots, d_N))\leq \mu(\text{U}(d_1, d_2d_3\dots d_N))=0.
\end{equation}
\end{corollary}
We have turned the question of the rareness of causality into a direct question about the measure theory of unitary groups, which is answered by corollary \ref{corr:causal_unitaries_measure_0}. The above result directly shows that almost all unitaries on a collection of spacelike separated subsystems are acausal, or equivalently the probability of picking a Haar random unitary and it being causal is $0$. 
\section{Results}\label{sec:results}
We wish to ask the same question for QFT. While corollary \ref{corr:causal_unitaries_measure_0} supports the intuition gleaned from QFT calculations, it cannot be directly generalised. Firstly, no Haar measure exists for infinite unitary groups, which are exactly those that arise in QFT. And secondly, there is no obvious choice of measure on the set of non-unitary quantum channels. In order to state similar results about the ``rarity" of causal channels in QFT, we require more sophisticated methods, based on topology rather than measure theory. Specifically, we show that causality picks out meagre (or even nowhere dense) subsets of local quantum channels. We work in the CB norm and weak$^*$-topology on the set of CB maps, as within these topologies the relevant sets of quantum channels are closed. 

The next few sections are dedicated to reviewing and developing the relevant Banach theory, and topological properties of quantum channels and completely bounded maps required to state our main results in section \ref{sec:rare}, relating the relative ``size" of the set of normal causal channels $\text{nCau}(\mathbf{K})$ within the set of normal channels that are local to $\mathbf{K}$, denoted $\text{nLoc}(\mathbf{K})$. Here, by locality, we mean that the channel acts trivially on all operators that are spacelike separated from $\mathbf{K}$.

\begin{restatable}{theorem}{nonunitarystuff}\label{thrm:causal_meagre}
Let $\mathfrak{A}$ be a $\sigma$-finite von Neumann QFT such that there exists at least one acausal channel in $\operatorname{nLoc}(\mathbf{K})$. For any compact subregion $\mathbf{K}$, the set of causal (normal) channels is nowhere dense in the set of local (normal) channels, i.e. $\operatorname{nCau}(\mathbf{K})$ is nowhere dense in $\operatorname{nLoc}(\mathbf{K})$ with respect to the $\operatorname{(}$CB-norm$\operatorname{)}$ weak$^*$, and pointwise weak and $\sigma$-weak topologies.
    
\end{restatable}

\noindent We take theorem $1$ as an affirmative answer to the main question posed. \textit{Causality is rare}, as both meagre and nowhere dense are strong statements about the ``rareness" of a subset. In fact, nowhere dense is sometimes referred to as \textit{rare}. To motivate the next few sections, we sketch the proof of theorem $1$. Firstly, we use the fact that a generalised form of Eq \eqref{eq:qm_causal} can equally be applied to any $\sigma$-weakly continuous completely bounded map. This allows us to view the causal channels as an intersection of \textit{local} channels with the kernel of a map from the completely bounded maps, defined by \eqref{eq:qm_causal}. After showing that the local channels are rare in the completely bounded maps, we use an inheritance property to finalise the proof.

\section{Banach and operator spaces}\label{sec:banach}
In order to answer our main question, we will study the following set of inclusions
\begin{equation}\label{eq:defining_inclusions}
    \text{nCau}(\mathbf{K})\subsetneq \text{nLoc}(\mathbf{K})\subsetneq \text{nUCP}(\mathfrak{A}(\mathbf{M}))\subsetneq \sigma \text{-CB}(\mathfrak{A}(\mathbf{M}))\subsetneq \text{CB}(\mathfrak{A}(\mathbf{M})),
\end{equation}
where $\text{CB}(\mathfrak{A}(\mathbf{M})), \sigma\text{-CB}(\mathfrak{A}(\mathbf{M}))$ are Banach spaces of completely bounded maps, and the remaining three subsets are convex (and to be defined later). This subsection recalls the properties of Banach spaces and completely bounded maps needed to define the sets above. We assume the reader is familiar with Banach spaces and the notion of a predual. 

In this section we recall the fundamental fact that the set of bounded maps between Banach spaces is itself a Banach space, and the set of completely bounded maps between two complete operator spaces is also a Banach space. Given a complex Banach space $V$, we denote by $M_{p,q}(V)$ the Banach space of $V$-valued $p\times q$ matrices, i.e. the elements $O=\{O_{ij}\}$. A map $\Psi:V\rightarrow V$ can be extended to a map $\Psi^{p,q}:M_{p,q}(V)\rightarrow M_{p,q}(V)$ by
\begin{equation}
    \{\Psi^{p,q}(O)\}_{i,j}=\Psi(O_{ij}).
\end{equation}
When $p=q=n$ we can equivalently denote this by $\Psi^n=\Psi\otimes \text{id}_n$, using the isomorphism $M_n(V)\cong V\otimes M_n(\mathbb{C})$. There exists a natural tower of matrix norms on the set of maps $\Psi:V\rightarrow V$,  $||\cdot||_n$, defined by
\begin{equation}
    ||\Psi||_{n}=||\Psi^n||=\sup_{\substack{v\in M_n(V),\\ ||v||_n\leq 1}} ||\Psi^n(v)||_n
\end{equation}

The supremum over all $n\geq 1$ defines the completely bounded norm $||\cdot||_\text{CB}$ on $\Psi$. If we consider the Banach space of bounded maps between von Neumann algebras, and its predual, we find $\text{B}(\mathfrak{A}, \mathfrak{B})_*\cong \mathfrak{A}\otimes_\pi\mathfrak{B}_*$ where $\otimes_\pi$ is the Banach space projective tensor product defined by the norm
\begin{equation}\label{eq:inf_b}
    ||u||_\pi=\inf\{\sum_i ||O_i||\;||\omega_i||\;| u=\sum O_i\otimes \omega_i, O_i\in \mathfrak{A}, \omega_i\in \mathfrak{B}_*\},
\end{equation}
where the infimum is understood as over all possible decompositions. In order to capture the completely bounded structure instead, we require the notion of operator spaces. An alternative viewpoint is motivated by noting the following equivalent form of the norm
\begin{equation}\label{eq:sup_b}
    ||u||_\pi=\sup \{|\phi(u)|\;|\phi\in (\mathfrak{A}\otimes \mathfrak{B}_*)^*, ||\phi||\leq 1\}.
\end{equation}
The norm required to instead pick out the completely bounded maps should satisfy
\begin{equation}\label{eq:sup_cb}
    ||u||_\wedge=\sup \{|\phi(u)|\;|\phi\in (\mathfrak{A}\otimes \mathfrak{B}_*)^*, ||\phi||_\text{CB}\leq 1\},
\end{equation}
where $||\phi||_{cb}$ is defined by the induced map $\mathfrak{A}\rightarrow \mathfrak{B}$. Looking over Eqs \eqref{eq:inf_b} to \eqref{eq:sup_cb}, it should not be too surprising that the definition of $||\cdot||_\wedge$ involves the matrix norms $M_p(\mathfrak{A}), M_q(\mathfrak{B}_*)$. With this motivation, we introduce operator spaces.

An operator space is a Banach space $V$ with a family of norms $||\cdot||_n, n\geq 1$, such that
\begin{enumerate}
    \item $||v\oplus w||_{n+m}=\max\{||v||_m, ||w||_n\}$,
    \item $||avb||_n \leq ||a||\;||v||_m||b||$
\end{enumerate}
for all $v\in M_m(V), w\in M_n(V)$. An operator space is complete if $M_n(V)$ are all complete. Von Neumann algebras $\mathfrak{A}\subseteq \mathcal{B}(\mathcal{H})$ are complete operator spaces with the obvious operator norms on $M_n(\mathfrak{A})$. Less obviously, the predual $\mathfrak{A}_*$ is also a complete operator space, see \cite{Effros2022-le}[page 44-45 and proposition 4.2.2]. A map $\varphi: V\rightarrow W$ between operator spaces is completely isometric if 
\begin{equation}
    \varphi_n:M_n(V)\rightarrow M_n(W)
\end{equation}
is an isometry for all $n$, and a complete isomorphism if $\varphi_n$ are isomorphisms for all $n$, and $||\varphi||_\text{CB}, ||\varphi^{-1}||_\text{CB}<\infty$. In fact, Eq \eqref{eq:cb_isomorphism} is an example of a complete isomorphism. The norm on $M_n(V\otimes W)$ given by 
\begin{equation}
    ||u||_\wedge=\inf_{p,q}\{||a||\,||b||\,||v||\,||w||\,|u=a(v\otimes w)b, \, a\in M_{n,pq}, b\in M_{pq,n}, v\in M_p(V), w\in M_q(W)\}
\end{equation}
where the infimum is over all $p,q$ and all decompositions $u=a(v\otimes w)b$. The space $V\widehat\otimes W$ is defined as the completion of the tensor product $V\circledcirc W$ of $V,W$ with respect to $||\cdot ||_\wedge$, and is a complete operator space if $V,W$ are too. We are ready to state a fundamental result in Banach and operator space theory.

\begin{proposition}\label{prop:completely_bounded_complete}
Let $(X,||\cdot ||_X)$ be a normed space, and $(Y, ||\cdot ||_Y)$ a Banach space. Then the set of bounded linear maps from $X$ to $Y$ is a Banach space with respect to the operator norm. If $X$ is an operator space and $Y$ a complete operator space, then the set of completely bounded linear maps from $X$ to $Y$ is a Banach space with respect to the completely bounded norm. 
\end{proposition}
The immediate application of the above result is to completely bounded maps between von Neumann algebras, which are complete operator spaces. 
\begin{corollary}
    Let $\mathfrak{A,B}$ be von Neumann algebras. The normed linear spaces $(\operatorname{B(\mathfrak A, \mathfrak B), ||\cdot||_1),\;}(\operatorname{CB}(\mathfrak{A}, \mathfrak{B}), ||\cdot||_\text{CB})$ are Banach spaces.
\end{corollary}
We note that the Banach dual space of $V$ is given by $\operatorname{CB(V, \mathbb{C}})$, as the completely bounded norm coincides with $||\cdot||$ for linear functionals, see \cite{Effros2022-le}[corollary 2.2.3]. Hence, we have
\begin{equation}\label{eq:cb_isomorphism}
    \operatorname{CB}(\mathfrak{A}, \mathfrak{B})\cong \operatorname{CB}(\mathfrak{A}, (\mathfrak{B}_*)^*)=\operatorname{CB}(\mathfrak{A}, \operatorname{CB}(\mathfrak{B}_*, \mathbb{C})),
\end{equation}
where we have used $\mathfrak{B}\cong (\mathfrak{B}_*)^*$. Then, by \cite{Effros2022-le}[proposition 7.1.2], for any operator spaces $V,W,X$ we have the following complete isomorphism
\begin{equation}
    \text{CB}(V\widehat\otimes W, X)\cong \text{CB}(V, \text{CB}(W,X)),
\end{equation}
and thus using Eq \eqref{eq:cb_isomorphism},
\begin{align}\label{eq:cb_predual}
     \operatorname{CB}(\mathfrak{A}, \mathfrak{B})\cong \text{CB}(\mathfrak{A}, \text{CB}(\mathfrak{B}_*, \mathbb{C}))\cong \operatorname{CB}(\mathfrak{A}\widehat{\otimes}\mathfrak{B}_*, \mathbb{C})=(\mathfrak{A}\widehat{\otimes}\mathfrak{B}_*)^*.
\end{align}
Having defined operator spaces, and recalled the predual of the set of completely bounded linear maps, we can continue to the topological properties of completely bounded maps and quantum channels.

\section{Quantum channels and topology}\label{sec:channels_topo}
As with the previous section, the following two subsections are required to define sets appearing in Eq \eqref{eq:defining_inclusions}, this time the middle and second from right. In subsection \ref{subsec:topo}, we recall the definitions of several important topologies in operator and Banach theory, allowing us to define $\sigma$-$\text{CB}(\mathfrak{A}, \mathfrak{B})$ while in subsection \ref{subsec:channels}, we recall the definition of a quantum channel, and define the convex set of normal quantum channels $\text{nUCP}(\mathfrak{A}, \mathfrak{B})$. 
\subsection{Topologies on algebras and maps}\label{subsec:topo}
A striking difference between infinite and finite dimensional operator algebras is the fact that operator topologies no longer coincide in infinite dimensions. This means that the seemingly harmless (in finite dimensions) statement:
\begin{equation}
    \braket{\psi|O_i\phi}\rightarrow \braket{\psi|O\phi}\; \forall \psi,\phi\in \mathcal{H}\implies O_i\ket{\xi}\rightarrow O\ket{\xi} \; \forall \xi \in \mathcal{H} 
\end{equation}
fails in general. Given a sequence (or net) of operators in a von Neumann algebra $O_i\in \mathfrak{A}\subseteq \mathcal{B}(\mathcal{H})$, it converges to $O$ in 
\begin{enumerate}
    \item norm topology if $||O_i-O||\rightarrow 0$,
    \item strong operator topology (SOT) if $O_i\ket{\psi}\rightarrow O\ket{\psi}$ for all $\psi\in \mathcal{H}$,
    \item and weak operator topology (WOT) if $\braket{\psi|O_i\phi}\rightarrow \braket{\psi|O\phi}$ for all $\psi, \phi\in \mathcal{H}$.
\end{enumerate}
We have listed the above topologies by decreasing strength in descending order, which can be confirmed by a simple exercise. The $\sigma$-weak (or ultraweak) topology on a von Neumann algebra, defines convergence of a net $O_i\rightarrow O$ if $\rho(O_i)\rightarrow \rho(O)$ for all normal functionals $\rho\in \mathfrak{A}_*$. It is weaker than norm topology, stronger than weak operator topology, and unordered with respect to strong operator topology.

Given a complex Banach space $V=(V_*)^*$, the weak$^*$-topology, denoted $\sigma(V,V_*)$ is the weakest topology on $V$ such that the map $T_v:V\rightarrow \mathbb{C}$
\begin{equation}
    T_v:\phi\mapsto \phi(x)
\end{equation}
are continuous. If we take $V=\mathfrak{A}_*, V^*=(\mathfrak{A}_*)^*\cong \mathfrak{A}$, then the weak$^*$ topology on $\mathfrak{A}$, denoted $\sigma( \mathfrak{A},\mathfrak{A}_*)$ is the weakest topology such that 
\begin{equation}
    T_\rho:O\mapsto \rho(O)
\end{equation}
is continuous for all $\rho$. Hence, it coincides with the $\sigma$-weak topology. A linear map $\Psi:\mathfrak{A}\rightarrow \mathfrak{B}$ is continuous in any of the above listed topologies if it sends convergent sequences (or nets when appropriate) to convergent sequences (or nets), e.g. $\Psi$ is $\sigma$-weakly (or weak$^*$) continuous if for all $O_i\rightarrow O$ $\sigma$-weakly, $\Psi(O_i)\rightarrow \Psi(O)$ $\sigma$-weakly. 

We will also use several topologies on spaces of linear maps. The operator topologies we have discussed induce pointwise operator topologies on linear maps: e.g. $\Psi_i:\mathfrak{A}\rightarrow \mathfrak{B}$ is pointwise convergent in $\sigma$-weak operator topology if $\rho(\Psi_i(O))\rightarrow \rho(\Psi(O))$ for all $\rho\in \mathfrak{B}_*, O\in \mathfrak{A}$. Two other topologies stand out for studying quantum channels, and completely bounded maps more generally. Let $\text{CB}(\mathfrak{A,B})$ be the Banach space of completely bounded maps. Then the norm topology $||O_i-O||_\text{CB}\rightarrow 0$ and weak$^*$ topology $\sigma(\text{CB}(\mathfrak{A,B}), \text{CB}(\mathfrak{A,B})_*)$ are of particular interest due to our interest in $\sigma$-continuous quantum channels. Firstly, however, we must show that $\sigma\text{-CB}(\mathfrak{A,B})$ is a Banach space. 

\begin{restatable}{proposition}{cbclosure}\label{prop:closed_sigma_cont}
    The set of $\sigma$-weakly continuous maps in $\operatorname{CB}(\mathfrak{A,B})$, denoted $\sigma\operatorname{-CB}(\mathfrak{A,B})$ is norm closed, and thus also a Banach space. 
\end{restatable}
\begin{proof}\label{proof:norm_closed}
        Let $\Psi_n\rightarrow \Psi$ be a CB-convergent sequence of $\sigma$-weakly continuous maps in $\text{CB}(\mathfrak{A,B})$, and $O_i\rightarrow O$ a $\sigma$-weakly convergent net in $\text{Ball}(\mathfrak{A})$. Consider any norm bounded set $S\subset \mathfrak{A}$, then
    
        \begin{equation}
            \sup_{O\in S}||\Psi_n(O)-\Psi(O)||\leq C||\Psi_n-\Psi||_\text{CB}\rightarrow 0,
        \end{equation}
        so $\Psi_n\rightarrow \Psi$ uniformly on a bounded set $S$. By assumption, the net $O_i$ is norm-bounded, and so $\Psi_n\rightarrow \Psi$ uniformly on the net. For any $\rho\in \mathfrak{B}_*$
        \begin{align}
            \lim_{i\rightarrow \infty}\rho(\Psi(O_i))=\lim_{i\rightarrow \infty}\lim_{n\rightarrow \infty} \rho(\Psi_n(O_i))=\lim_{n\rightarrow \infty}\lim_{i\rightarrow \infty} \rho(\Psi_n(O_i))=\lim_{n\rightarrow \infty} \rho(\Psi_n(O))= \rho(\Psi(O)),
        \end{align}
        where we have swapped the limits using uniform convergence (Moore-Osgood), then used that $\Psi_n$ are $\sigma$-weakly continuous for the second equality, and that CB-convergence implies pointwise $\sigma$-weak convergence for the last equality. This implies that $\rho\circ \Psi:\text{Ball}(\mathfrak{A})\rightarrow \mathbb{C}$ is $\sigma$-weakly continuous, and thus $\sigma$-weakly continuous when extended uniquely to $\mathfrak{A}\rightarrow \mathbb{C}$. Since $\rho\in \mathfrak{B}_*$ is generic, $\Psi\in \sigma\text{-CB}(\mathfrak{A,B})$.  Linearity is immediate, and so $\sigma\text{-CB}(\mathfrak{A,B})$ is Banach. 
    \end{proof}

With the above result, we see that $\sigma\operatorname{-CB}(\mathfrak{A,B})\subseteq \operatorname{CB}(\mathfrak{A,B})$ is a closed subspace. We now compare the $\sigma(\operatorname{CB}(\mathfrak{A}, \mathfrak B), \operatorname{CB}(\mathfrak{A}, \mathfrak B)_*)$ and pointwise topologies. 

\begin{proposition}\label{prop:all_topo}
    Let $\mathfrak A, \mathfrak B$ be von Neumann algebras. Then on $\operatorname{CB}(\mathfrak{A},\mathfrak B)$
    \begin{enumerate}
        \item we have the following ordering of topologies: pointwise weak $\leq $ pointwise $\sigma$-weak$\leq \sigma(\text{CB}(\mathfrak{A}, \mathfrak B), \text{CB}(\mathfrak{A}, \mathfrak B)_*)\leq$ CB norm. 
        \item on bounded subsets of $\operatorname{CB}(\mathfrak{A},\mathfrak B)$, the weakest three topologies coincide.
    \end{enumerate}
    \begin{proof}
        We start with statement $1$. It suffices to show that the respective weaker (semi)-norms are implied by the stronger ones.
        Suppose $\Psi_i\in \text{CB}(\mathfrak A, \mathfrak B)$ is a CB norm convergent net. By the definition of the dual norm, we have
        \begin{equation}
            ||\Omega||_\wedge||\Psi-\Psi_i||_\text{CB}=||\Omega||_\wedge\sup_{\Omega'\in \mathfrak{A}\widehat{\otimes}\mathfrak{B}_*, \;||\Omega'||_\wedge\leq 1}|\braket{ \Psi-\Psi_i, \Omega'}|\;  \geq |\braket{ \Psi-\Psi_i, \Omega}|
        \end{equation}
        for any $\Omega\in \mathfrak{A}\widehat{\otimes}\mathfrak{B}_*$. Hence, CB norm convergence implies weak$^*$. Since $\mathfrak{A}\circledcirc\mathfrak B_*\subseteq \mathfrak A\widehat\otimes \mathfrak B_*$, convergence under the seminorms $|\braket{O\otimes \rho, }|$ are implied by weak$^*$-convergence. Hence, pointwise $\sigma$-weak convergence follows from weak$^*$-convergence. Letting $\rho(\cdot)=\text{tr}(\ket{\psi}\bra{\phi}\cdot)$, it is clear that pointwise weak convergence is implied by pointwise $\sigma$-weak convergence. 

        We now move onto statement $2$. Let $S\subset \text{CB}(\mathfrak A, \mathfrak B)$ be CB norm bounded by $C$ and $\Psi_i$ a pointwise weakly convergent net in $S$. Let $\rho\in \mathfrak B_*$. Thanks to their norm density in $\mathcal B(\mathcal H)_*/\perp\cong \mathfrak B_*$, the following nets 
        \begin{equation}
            \rho_n(\cdot)=\text{tr}\left( \sum_{j=1}^{N_n} \ket{\psi^n_j}\bra{\phi^n_j}\cdot\right).
        \end{equation}
         can approximate $\rho$ in norm topology. Let $a_{i,n}=\braket{\Psi_i, O\otimes \rho_n}$ for any $O\in \mathfrak A$, then
         \begin{equation}
             |a_{i,n}-a_i|=|\braket{\Psi_i, O\otimes \rho_n-O\otimes \rho}|\leq C||O\otimes \rho_n -O\otimes \rho||_\wedge.
         \end{equation}

         Then
         \begin{align}
             \rho(\Psi(O))=\lim_{n\rightarrow \infty}\rho_n(\Psi(O))=\lim_{n\rightarrow \infty}\lim_{i\rightarrow \infty} \rho_n(\Psi_i(O))=\lim_{i\rightarrow \infty}\lim_{n\rightarrow \infty} \rho_n(\Psi_i(O))=\lim_{i\rightarrow \infty}\rho(\Psi_i(O)),
         \end{align}
         where the second equality follows from pointwise weak convergence and continuity, and the third equality follows by applying Moore–Osgood. Hence, $\Psi_i\rightarrow \Psi$ in pointwise $\sigma$-weak topology. 

         Now we assume that $\Psi_i \in S$ is pointwise $\sigma$-weakly convergent to $\Psi$. Pick $\Omega\in \mathfrak A\widehat \otimes \mathfrak B_*$ and a net $\Omega_n\in \mathfrak A\circledcirc\mathfrak B_*$, such that $||\Omega_n-\Omega||_\wedge\rightarrow 0$. Define $b_{i,n}:=\braket{\Psi_i, \Omega_n}$ and note that $|b_{i,n}-b_i|=|\braket{\Psi_i, \Omega_n-\Omega}|\leq C||\Omega_n-\Omega||_\wedge$. Then, as before, we see that 
         \begin{equation}
             \lim_{i\rightarrow \infty}\braket{\Psi_i, \Omega}=\braket{\Psi, \Omega}.
         \end{equation}
    \end{proof}
\end{proposition}

\subsection{Quantum channels}\label{subsec:channels}
Quantum channels are considered the most general evolution in quantum theory thanks to Stinespring's theorem, and have an elegant operator algebraic theory. Here, we recall a few elements of this theory.

A map $\Psi: \mathfrak{A}\rightarrow \mathfrak{B}$ between von Neumann algebras $\mathfrak{A,B}$ is positive if it maps positive elements to positive elements $\Psi(\mathfrak{A}^+)\subseteq \mathfrak{B}^+$. A map is completely positive if $\Psi^n$ is positive for all $n\geq 1$. A (always from here on Heisenberg) quantum channel is a linear map between von Neumann algebras $\Phi:\mathfrak{A}\rightarrow \mathfrak{B}$ that is unital $\Phi(1)=1$ and completely positive. We denote the set of such maps by $\text{UCP}(\mathfrak{A, \mathfrak{B}})$ which sits in $\text{CB}(\mathfrak{A}, \mathfrak{B})$ as a convex set:
\begin{equation}
    p\Phi_1+(1-p)\Phi_2\in \text{UCP}(\mathfrak{A}, \mathfrak{B})
\end{equation}
for all $\Phi_i\in \text{UCP}(\mathfrak{A}, \mathfrak{B}), p\in [0,1]$. A linear map $\Psi:\mathfrak{A}\rightarrow \mathfrak{B}$ is $\sigma$-weak continuous (or ultraweakly continuous) if and only it preserves the algebraic predual, i.e. if $\omega\circ\Psi\in \mathfrak{A}_*$ for all $\omega\in\mathfrak{B}_*$. Note that this coincides with the weak$^*$ topology on $\mathfrak{A}$. For positive maps, it also coincides with normality, which requires a map to satisfy $\sup_i\Psi(O_i)=\Psi(\sup_iO_i)$ for any increasing net $O_i$ of operators in $\mathfrak{A}$. We denote the subset of normal UCP maps with $\text{nUCP}$. All maps in $\text{nUCP}(\mathcal{B}(\mathcal H), \mathcal  {B}(\mathcal{K}))$ can be written in terms of a countable number of Kraus operators $K_i\in \mathcal{B}(\mathcal{K}, \mathcal{H})$ with 
\begin{equation}
    \Phi(O)=\sum_i K_i^\dagger OK_i, \;\; \sum_i K_i^\dagger K_i=1,
\end{equation}
where convergence is understood in strong operator topology. The question of extensions of maps from $\mathfrak A\subsetneq \mathcal{B}(\mathcal H)$ in the context of QFT, and locality of those extensions in quantum channels, is discussed in e.g. \cite{redei_how_2010,redei_when_2010}. We will not use the Kraus operator decomposition in our proofs, and so questions of existence when $\mathfrak{B}\subsetneq \mathcal{B}(\mathcal{K})$ are avoided.

\section{Closure and compactness}\label{sec:closure_compact}
Let $\mathcal{H, K}$ be separable Hilbert spaces, and for any von Neumann algebras $\mathfrak{A}\subseteq \mathcal{B}(\mathcal{H}), \mathfrak{B}\subseteq \mathcal{B}(\mathcal{K})$, denote the set of unital completely positive maps from $\mathfrak{A}$ to $\mathfrak{B}$ by $\text{UCP}(\mathfrak{A,B})$, and the normal UCP maps by $\text{nUCP}(\mathfrak{A,B})$. 

In order to answer questions about closure and compactness, we need to consider the ambient spaces quite carefully. Given a von Neumann algebra $\mathfrak{A}$, we denote the set of all linear maps as $\mathcal{L}(\mathfrak{A})$. By proposition \ref{prop:completely_bounded_complete}, $\text{CB}(\mathfrak{A})$ is a Banach space. Importantly, we can show that the set of all $\sigma$-weakly continuous maps in $\text{CB}(\mathcal{A})$ is also Banach. We do not claim any novelty in these results: result $2$ is stated without proof in \cite{Farenick_2013}, and related to, but different from, \cite{van2025pure}[remark 12]. The conclusion in statement $3$ can be reached from the results in \cite{kretschmann2007continuitytheoremstinespringsdilation}, however it does not appear explicitly there. 

\begin{restatable}{proposition}{test}\label{prop:closed_compactness_boundedness}
    For any von Neumann algebras $\mathfrak{A}\subseteq \mathcal{B}(\mathcal{H}), \mathfrak{B}\subseteq \mathcal{B}(\mathcal{K})$,
    \begin{enumerate}
        \item $\operatorname{UCP}(\mathfrak{A},\mathfrak{B})$ is compact, and closed in $\operatorname{CB}(\mathfrak{A, B})$ in the weak$^*$ topology on $\operatorname{CB}(\mathfrak{A, B})$. The same holds for pointwise weak and $\sigma$-weak topologies.
        \item $\operatorname{nUCP}(\mathfrak{A,B})$ is not pointwise $\sigma$-weakly closed in $\operatorname{UCP}(\mathfrak{A,B})$ (or $\mathcal{L}(\mathfrak{A,B})$) if $\mathfrak{A}$ is infinite dimensional.
        \item $\operatorname{(n)UCP}(\mathfrak{A,B})$ is CB-norm closed and uniformly CB-bounded in $(\sigma\operatorname{-)CB}(\mathfrak{A, B})$. 
    \end{enumerate}
\end{restatable}
\begin{proof}\label{proof:compactnesse_etc}
    \begin{enumerate}
    \item The weak$^*$ topology on $\text{CB}(\mathfrak{A,B})$ is given by the elements $\phi\in \text{CB}(\mathfrak{A,B})_*=\mathfrak {A}\widehat\otimes \mathfrak B_*$, where $\phi=\sum_i O_i\otimes \omega_i$, and $O_i\in \mathfrak{A}, \omega_i\in \mathfrak{B}_*$. 
    The Banach-Alaoglu theorem states that the unit ball in $\text{CB}(\mathfrak{A}, \mathfrak B)$ is weak$^*$ compact, i.e. in $\sigma(\text{CB}(\mathfrak{A,B}), \mathfrak {A}\widehat\otimes \mathfrak B_*)$. Since for CP maps,
            \begin{equation}\label{eq:uniform_bounded_ucp}
                ||\Phi||_\text{CB}=||\Phi||=1,
            \end{equation}
            the set $\operatorname{UCP}(\mathfrak{A,B})$ is contained in the compact ball. It follows that if $\operatorname{UCP(\mathfrak{A,B})}$ is weak$^*$-closed, then it is weak$^*$-compact. 

            We prove that statement now. Let $\Phi_n$ be a weak$^*$- convergent net in $\text{UCP}(\mathfrak{A,B})$ with limit $\Phi$. Unitality is obvious, we start with complete positivity. Since $\Phi$ is completely bounded and unital, it is completely positive if and only if it is contractive, i.e. $||\Phi||_\text{CB}\leq 1$. Since we have seen that $\text{UCP}(\mathfrak{A,B})$ is a subset of the (weak$^*$-closed) unit ball, it follows that $\Phi$ is UCP.

            By prop \ref{prop:all_topo}, the claim also holds for pointwise and $\sigma$-weak topologies.
            \item Recall that $\mathfrak{A}_*$ is not weak$^*$-closed in $\mathfrak{A}^*$. Let $\omega_n\in \mathfrak{A}_*$ be a net of normal states that converges in weak$^*$ topology to a non-normal state $\omega$. The linear maps $\Psi_n:O\mapsto \omega_n(O) I$ form a net in $\operatorname{nUCP}(\mathfrak{A,B})$. Then, since 
        \begin{equation}
            \rho(\Psi_n(O)-\Psi(O))=\rho(I)(\omega_n(O)-\omega(O))\rightarrow 0,
        \end{equation}
        $\Psi_n$ converges pointwise $\sigma$-weakly to $\Psi$, $\operatorname{nUCP}(\mathfrak{A})$ is not $\sigma$-weakly closed.
        \item Uniform boundedness comes from Eq \eqref{eq:uniform_bounded_ucp}. To see closure in CB-norm, let $\Phi_k\in \text{nUCP}(\mathfrak{A,B})$ converge in CB-norm to $\Phi$. Then, for all $n\geq 1$
        \begin{equation}
            ||(\Phi_k-\Phi)\otimes \text{id}_n||\leq ||\Phi_k-\Phi||_\text{CB},
        \end{equation}
        and so $\Phi$ is UCP. Finally, we have the following inequality 
        \begin{equation}
            ||\omega\circ (\Phi_k-\Phi)||\leq ||\omega||\,||\Phi_k-\Phi||\leq ||\omega||\,||\Phi_k-\Phi||_\text{CB}.
        \end{equation}
        Hence, $\mathfrak{A}_*\ni\omega\circ\Phi_k\rightarrow \omega\circ\Phi\in \mathfrak{A}^*$. By the norm closure of $\mathfrak{A}_*$ in $\mathfrak{A}^*$, $\omega\circ\Phi$ is normal for all $\omega\in \mathfrak{B}_*$. Hence, $\Phi$ is normal and in $\text{nUCP}(\mathfrak{A,B})$. A similar proof holds in the non-normal case. 
        \end{enumerate}

    \end{proof}
    
The above proposition shows that, despite the importance of normality/$\sigma$-weak continuity for defining meaningful quantum channels, pointwise $\sigma$-weak continuity is not a helpful topology in which to study the structure of normal quantum channels. Both CB and weak$^*$-topologies lead to closed sets of quantum channels, which will allow us to state rarity results in the later sections. 

As a final result for this section, we prove that unitary channels, which are always normal, are closed in CB norm. 
\begin{proposition}\label{prop:closure_unitary}
    The set of unitary channels $\operatorname{U}(\mathbf K)$ is CB norm closed in the set of normal channels $\operatorname{nUCP}(\mathbf K)$. 
\end{proposition}
\begin{proof}
     Let $\text{U}(\mathbf K)\ni\Phi_i\rightarrow \Phi$ in CB norm, such that $\Phi_i$ are all unitary. By the closure of normal channels under CB norm, it follows that $\Phi\in \text{nUCP}(\mathbf K)$. By \cite[Theorem 1]{kretschmann2007continuitytheoremstinespringsdilation}, for any two channels $\Psi_i$, their Stinespring isometries $V_i:\mathcal H\rightarrow \mathcal K$ obey
     \begin{equation}\label{eq:op_bound}
         \inf_{V_1,V_2}||V_1-V_2||\leq \sqrt{||\Psi_1-\Psi_2||_\text{CB}}
     \end{equation}
     Since the channels are normal it follows that we can pick $V_i:\mathcal H\rightarrow \mathcal H\otimes \mathcal J$ where $\mathcal J$ is separable, and 
     \begin{equation}
         \Psi_i(\cdot)=V_i(\cdot \otimes 1)V_i^\dagger.
     \end{equation}
     Taking $\ket{e_i}\in \mathcal J$ to be an orthonormal basis, we can define $K^i_j=(1\otimes \bra{e_j})V_i$. A sufficient and necessary condition for a channel to be unitary is that for all $i,j$,
     \begin{equation}
         K_i^n{}^\dagger K_j^n, K_i^n K_j^n{}^\dagger\in \mathbb C1.
     \end{equation}
     From the operator bound Eq \eqref{eq:op_bound}, we have
     \begin{equation}
         ||K_i-K^n_i||=||(1\otimes \bra{e_j})(V-V^n)||\leq C_j ||V-V^n||\rightarrow 0.
     \end{equation}
    It then follows that for any $O\in \mathcal B(\mathcal H)$
     \begin{equation}
          [K_i^\dagger K_j,O]=[\lim_{n }K_i^n{}^\dagger K_j^n, O]=\lim_{n }[K_i^n{}^\dagger K_j^n, O]=0,
     \end{equation}
     since the commutator is norm continuous, and likewise for $K_i K_j^\dagger$. Since $O$ is generic, both $K_i^\dagger K_j, K_i K_j^\dagger\in \mathbb C1$, and so $\Phi$ is unitary. 
\end{proof}

\section{QFTs and causality}\label{sec:qft}
A natural framework for discussing causality is the algebraic approach to QFT (AQFT). An alternative approach is based on the general boundary formalism  \cite{oeckl_local_2019, oeckl2024spectraldecompositionfieldoperators, oeckl2025causalmeasurementquantumfield, oeckl2025localcompositionalmeasurementsquantum}, which provides a contrasting and complementary way to study the causality of operations. 

For our purposes, the following broad definition of an AQFT is sufficient. 

\begin{definition}
    A QFT is a map (functor) $\mathfrak{A}$ from the set of open, causally convex subsets of a globally hyperbolic manifold $\mathbf{M}$ to the non-Abelian von Neumann subalgebras of $\mathcal{B}(\mathcal{H})$ e.g. $\mathbf{R}\mapsto \mathfrak{A}(\mathbf{R})\subseteq \mathcal{B}(\mathcal{H})$, such that 
\begin{enumerate}
    \item if $\mathbf{R}\subseteq \mathbf{S}$, then $\mathfrak{A}(\mathbf{R})\subseteq \mathfrak{A}(\mathbf{S})$.
    \item if $\mathbf{S}, \mathbf{R}$ are spacelike separated, then $[\mathfrak{A}(\mathbf{S}), \mathfrak{A}(\mathbf{R})]=0$. 
\end{enumerate}
\end{definition}

The largest algebra of interest is $\mathfrak{A}(\mathbf{M})$, which represents the set of all observables associated to the QFT. In the context of a QFT, a quantum channel is an element of $\text{UCP}(\mathfrak{A}(\mathbf{M}))=:\text{UCP}(\mathbf{M})$. Given some region $\mathbf{K}$, we denote its in and out regions by $\mathbf{K}_{\text{in/out}}\mathbf{M}\backslash (J^{+/-}(\mathbf{K}))$, where $J^\pm(\mathbf{K})$ are the set of all points to the past or future of $\mathbf{K}$ respectively. Finally, we denote spacelike separated complement $\mathbf{K}^\perp$, i.e. the set of all points that are spacelike separated from $\mathbf{K}$.

The local channels on a region $\mathbf{K}$ are the channels that act trivially on its spacelike complement, generalising the notion of a local channel, $\Phi_1\otimes \mathbb{I}_2$, from the finite dimensional model in section \ref{sec:motivation}.
\begin{definition}
    A channel $\Phi$ local to a region $\mathbf{K}$ has 
    \begin{equation}
        \Phi|_{\mathbf{K}^\perp}=\operatorname{id}.
    \end{equation}
\end{definition}
The set of local channels forms a convex set $\text{Loc}(\mathbf{K})\subseteq \text{UCP}(\mathbf{M})$, however in general the causal ordering of operations means that we only \textit{apply} the channel to subalgebras that correspond to regions that are contained in $\mathbf{K}_\text{out}$. The set of normal local channels is exactly the set of normal channels $\Phi$ such that 
\begin{equation}
    \Phi|_{\mathbf{K}^\perp}=\text{id},
\end{equation}
 and are denoted by $\text{nLoc}(\mathbf{K})\subseteq \text{nUCP}(\mathbf{M})$, see \cite{van2025pure, ruep_thesis_2022}. The normal local channels are of particular interest to us, as we explain at the end of this section, and so we show that they are closed in CB topology. These results, while quite direct, do not appear to have been stated before. We draw attention again to \cite{van2025pure}[remark 12], which provides a counterexample for closure in the pointwise ultraweak topology. 

\begin{proposition}
    Let $\mathfrak{A}$ be a QFT such that every local region is assigned a von Neumann algebra. The set of all local channels is closed in weak$^*$-topologies, and the set of all normal local channels is closed in $\operatorname{CB}$-norm. 

    \begin{proof}
       To see closure of the normal local channels in CB norm topology, we need only show that locality is preserved. Pick $O\in \mathfrak{A}(\mathbf{K}^\perp), ||O||\leq 1$, and a net $\Phi_i\in \text{nLoc}(\mathbf{K})$ that converges in CB norm to $\Phi\in \text{nUCP}(\mathbf{M})$. Then
        \begin{equation}
            ||\Phi(O)-O||_1=||(\Phi-\Phi_i)(O)||_1\leq ||\Phi-\Phi_i||_\text{CB}||O||_1\rightarrow 0
        \end{equation}
        so $\Phi\in \text{nLoc}(\mathbf{K})$. 
        
        Finally, if $\text{Loc}(\mathbf{K})\ni\Phi_i\rightarrow \Phi$ in weak$^*$ topology, then for all $O\in \mathfrak{A}(\mathbf{K}^\perp)$,
        \begin{equation}
            \rho(O-\Phi(O))=\rho(\Phi_i(O)-\Phi(O))\rightarrow 0.
        \end{equation}
        Hence, $\Phi\in \text{Loc}(\mathbf{K})$ and we have closure of local channels in weak$^*$-topology. 
    \end{proof}
\end{proposition}

Sorkin's seminal paper \cite{sorkin1993impossible} noted that the set of local channels is strictly larger than the set of causal channels. Instead, causal channels $\Phi$ local to $\mathbf{K}$ must also satisfy a generalisation of Eq \ref{eq:qm_causal}.

\begin{definition}
    A channel local to $\mathbf{K}$ is causal if \begin{equation}
    \rho\left(\Phi_\mathbf{S}(\Phi(O_\mathbf{R}))-\Phi(O_\mathbf{R})\right)=0,
\end{equation}
for all states $\rho\in \mathfrak{A}(\mathbf{M})^*$, all spacelike separated compact regions $\mathbf{S}, \mathbf{R},\mathbf{S}\in \mathbf{K}_\text{in},\mathbf{R}\in \mathbf{K}_\text{out}$, all operators $O_\mathbf{R}\in\mathfrak{A}(\mathbf{R})$, and all local channels $\Phi_\mathbf{S}\in \text{Loc}(\mathbf{S})$, see Fig \ref{fig} for the spacetime setup. 
\end{definition}
This is very similar to the case we looked at in section \ref{sec:motivation}, however we do not assume all systems are spacelike separated. Operationally, we can view $\rho$ as an initial state, $\Phi_\mathbf{S}$ as a preparation, and $O_\mathbf{R}$ as a measurement. Since the preparation and measurement are spacelike separated events, causality requires that the preparation cannot be detected, regardless of an intermediate intervention $\Phi$. 
\begin{figure}
    \centering
    \includegraphics[width=0.5\linewidth]{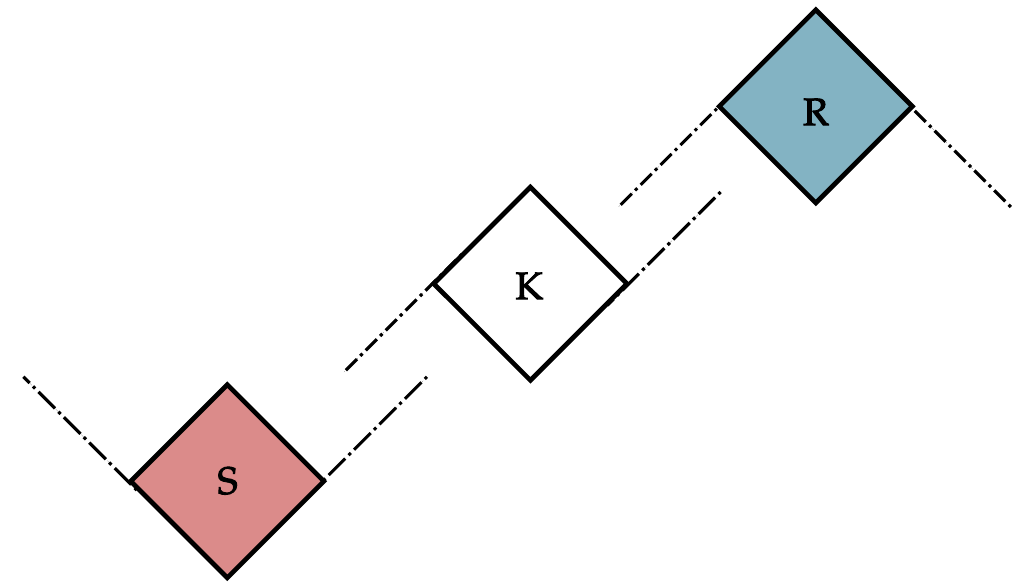}
    \caption{The spacetime setup of a Sorkin scenario is shown, where $\mathbf{S}$ is before $\mathbf{K}$, which is before $\mathbf{R}$, yet $\mathbf{S}$ and $\mathbf{R}$ are spacelike separated. The lightcones are indicated by dashed lines.}
    \label{fig}
\end{figure}
We denote the set of channels local to $\mathbf{K}$ (and normal)  that are causal (and normal) by $\text{(n)Cau}(\mathbf{K})$. 

Since we have been careful to distinguish normal states and channels, there is an important question about the robustness of Sorkin's causality condition. Specifically, we show that it is equivalent to testing the causality of a channel $\Phi$ using only normal preparations and initial states. 

\begin{restatable}{proposition}{causalforall}
\label{prop:equiv_of_conditions}
    Let $\mathfrak{A}$ be a QFT such that every local region is a von Neumann algebra. Then the following two statements are equivalent
    \begin{enumerate}
        \item $\Phi$ is causal for all normal states $\rho$, all normal channels $\Phi_\mathbf{S}$, and all $O_\mathbf{R}$.
        \item  $\Phi$ is causal for all states $\rho$, all channels $\Phi_\mathbf{S}$, and all $O_\mathbf{R}$.
    \end{enumerate}
\end{restatable}
\begin{proof}\label{proof:qft_equiv}
        We note that $(2)\implies (1)$ is trivial. For the converse direction, we start by assuming that $\Phi$ is causal for all normal states $\rho$, all normal channels $\Phi_\mathbf{S}$, and all $O_\mathbf{R}$. 
        
        The normal states of a von Neumann algebra are separating, so $\rho(O)=0$ for all normal states if and only if $O=0$. Hence, if 
        \begin{equation}
            \rho\left(\Phi_\mathbf{S}(\Phi(O_\mathbf{R}))-\Phi(O_\mathbf{R})\right)=0,
        \end{equation}
        for all normal $\rho$, then $\Phi_\mathbf{S}(\Phi(O_\mathbf{R}))=\Phi(O_\mathbf{R})$. Demanding causality for all normal states implies $\Phi_\mathbf{S}(\Phi(O_\mathbf{R}))=\Phi(O_\mathbf{R})$ for all normal channels $\Phi_\mathbf{S}$ and all $O_\mathbf{R}$. 

        Demanding causality with respect to all normal channels $\Phi_S$ and all operators $O_\mathbf{R}$ includes demanding causality with respect to all unitary channels $\Phi_S(\cdot)=U^\dagger \cdot U$ with $U\in \mathfrak{A}(\mathbf{S}^\perp)'$, and all self-adjoint operators $O_\mathbf{R}$. Hence, \begin{equation}\label{eq:commutation_condition}
            U^\dagger \Phi(O_\mathbf{R})U=\Phi(O_\mathbf{R}) , \;\forall U, O_\mathbf{R}  \iff  [U, \Phi(O_\mathbf{R})]=0\;\forall U, O_\mathbf{R} \iff  [\mathfrak{A}(\mathbf{S^\perp})', \Phi(\mathfrak{A}(\mathbf{R}))]=0,
        \end{equation}
        where we have used that both the unitaries and the self adjoint operators can be used to linearly span von Neumann algebras. The above commutation relation implies that $\Phi(\mathfrak{A}(\mathbf{R}))\subseteq \mathfrak{A}(\mathbf{S}^\perp)$. Hence, for any $\Phi_\mathbf{S}\in \text{Loc}(\mathbf{S})$,
        \begin{equation}\label{eq:operator}        \Phi_\mathbf{S}(\Phi(O_\mathbf{R}))=\Phi(O_\mathbf{R}).
        \end{equation}
        From this final expression, $(2)$ follows immediately. 
        \end{proof}

While there are good reasons to believe that normal channels are the only ones that appear in QFT---both mathematical (the FV framework predicts normal channels)\cite{simmons_brukner}, and conceptual (because only normal channels have a Schrödinger picture)---the above result provides some reassurance that we are not being overly specific. We could imagine that some channels could be ruled out by testing against non-normal preparation channels and non-normal states. However, we have shown in proposition \ref{prop:equiv_of_conditions} that no such problem occurs. If a channel is causal with respect to normal preparations and states (statement $1$ in proposition \ref{prop:equiv_of_conditions}), it is causal with respect to all preparations and states (statement $2$ of \ref{prop:equiv_of_conditions}), and vice-versa.

\section{Causality is rare in QFT}\label{sec:rare}
We are now ready to return to our original question: \textit{how rare is causality?} As we have discussed, it is tempting to want to apply the results of section \ref{sec:motivation} directly to QFT, however the local and causal unitary groups are not locally compact (in any useful topology) and so Steinhaus-Weil does not apply, nor does a Haar measure exist. Instead, we recall some notions from topology and descriptive set theory that provide a precise definition of when a subset is ``rare". For a complete introduction, see \cite{Kechris2012-jm}, particularly sections $1.8,1.9$.

A subset $S\subset X$ of a topological space $X$ is described as nowhere dense, or \textit{rare}, if its closure has an empty interior i.e. $\text{int}_X (\overline{S}^X)=\emptyset$. An example is the integers as a subset of the reals, every integer is entirely separated from every other integer. A meagre set $\Sigma$ is a countable union $\Sigma=\bigcup_{i\in I} S_i$ of nowhere dense sets $S_i$, and so the rational numbers are a meagre set in the reals. Given two subsets $S,T\subseteq X$, the symmetric difference is $S\Delta T:=S\cup T\backslash S\cap T$. We say a subset $S$ has the Baire property if its symmetric  difference from an open set is a meagre set, i.e. if $S\Delta U$ is a meagre set for some open set $U$. 

Since these definitions are not in common usage within the quantum information community, we collect a few standard properties of nowhere dense sets, meagre sets, and sets with the Baire property. While it is enough for our purposes to simply state these results, we provide proofs for completeness.
\begin{proposition}\label{prop:descriptive_set_theory}

    \begin{enumerate}
        \item The boundary of any closed subset $S\subseteq X$ is nowhere dense in $X$.
        \item Any closed subset has the Baire property.
        \item Any strict closed subspace of a normed vector space $V$ is nowhere dense in $V$.
        \item Let $X$ be a topological vector space, $L\subset X$ a closed linear subspace, and $C\subset X$ a convex subset. If $C\not\subset L$, then $C\cap L$ has empty relative interior in $C$. Hence $C\cap L$ is nowhere dense in $C$, provided it is closed in $C$.
    \end{enumerate}

    \begin{proof}
        \begin{enumerate}
            \item Let $U$ be an open subset of $\partial S=S\backslash \text{int}_X(S)$. Then, $U\subseteq \text{int}_X(S)$ as it is open, and $\partial S\subseteq S$ by closure. Hence, $U\subseteq \text{int}_X(S)\cap \partial S=\emptyset$ by definition. We can conclude that $\text{int}_X(\partial S)=\text{int}_X(\overline{\partial S}^X)=\emptyset$.
            \item Let $S\subseteq X$ be a closed subset. Then $\partial S:=S\backslash \text{int}_X (S)=(S\cup \text{int}_X (S))\backslash (S\cap \text{int}_X (S))=\text{int}_X (S)\Delta S$. Since $\text{int}_X (S)$ is open, there exists an open set such that the symmetric difference of it with $S$ is meagre by part $1$.
            \item Let $V$ be a normed vector space, and $W\subsetneq V$ be a closed strict subspace. Since $W$ is closed, it is nowhere dense if 
            \begin{equation}
                \text{int}_V(\overline{W}^V)=\text{int}_V(W)=\emptyset
            \end{equation}
            Assuming that $W$ is not nowhere dense, we can find an open ball $B(w,r)\subset W$ for some $w\in \text{int}_X(W)\subseteq W, r>0$. However, we also have that $B(w,r)-w=B(0,r)\subset W$. By linearity, any subspace containing a neighbourhood of the origin is $V$, so $W=V$, in contradiction with the assumption that it is a strict subspace. Hence, $W$ is nowhere dense.
            \item Suppose $x\in C\cap L$ is a relative interior point. Then there exists an open $U\subset X$ such that $x\in U, x\in U\cap C\subset L$. Since $C\not\subset L$, let $y\in C\backslash L$ and define
            \begin{equation}
                x_t=x(1-t)+ty\in C.
            \end{equation}
            For some $t'>0$, we have $x_{t'}\in U$, so $x_{t'}\in U\cap C \subset L$. However, by linearity, 
            \begin{equation}
                t'y=x_{t'}-(1-t')x\in L,\;  y\in L
            \end{equation}
            which is a contradiction. 
        \end{enumerate}
    \end{proof}
\end{proposition}

We have now, finally, developed and reviewed enough tools to state and prove our main result, which demonstrates that the set of causal (normal) channels local to some region is a rare subset of all (normal) channels local to that region.

\nonunitarystuff*
\begin{proof}
    Consider $X=\text{CB}(\mathbf{M})$ or $\sigma$-$\text{CB}(\mathbf{M})$, equipped with either weak$^*$ or CB norm topology respectively, and the appropriate subspace $C=\text{Loc}(\mathbf{K})$ or $\text{nLoc}(\mathbf{K})$ respectively. It is easy to verify that  $C\subseteq X$ is a convex set. 

    Next, we consider the function
    \begin{equation}\label{eq:causality_constraint_CB}
    \Gamma_{\Phi_\mathbf{S},O_\mathbf{R}}^\rho:\Psi\mapsto \rho(\Phi_\mathbf{S}(\Psi(O_\mathbf{R}))-\Psi(O_\mathbf{R})).
        \end{equation}
        and 
        \begin{equation}
            L:=\bigcap_{\substack{\rho\in \mathfrak{A}(\mathbf{M})_*}}\bigcap_{\substack{\mathbf{S}\perp \mathbf{R} ,\\
            \mathbf{S}\in \mathbf{K}_\text{in},\\ \mathbf{R}\in \mathbf{K}_\text{out}}}\bigcap_{\substack{\Phi_\mathbf{S}\in \text{Loc}(\mathbf{S}), \\O_\mathbf{R}\in \mathfrak{A}(\mathbf{R})}}\text{ker }\Gamma^\rho_{\Phi_\mathbf{S},O_\mathbf{R}},
        \end{equation}
    which is clearly a linear space. The space of (normal) causal channels is given by $C\cap L$. By prop \ref{prop:descriptive_set_theory}.4, it is sufficient to show that $L$ is closed and a proper subset. Since $L$ is an intersection of kernels, it follows that it is closed if $\Gamma^\rho_{\Phi_\mathbf{S},O_\mathbf{R}}$ is continuous in the respective topologies, which is easy to verify\footnote{In the case of normal channels, we must restrict to normal $\rho, \Phi_\mathbf S$.}. Since there exists an acausal channel by assumption, $L$ is closed and strict, and so the set of (normal) causal quantum channels is nowhere dense in the set of (normal) local quantum channels in the (CB norm) weak$^*$ topology. 

    Finally, we recall prop \ref{prop:all_topo} and the fact that $||\Psi||_\text{CB}\leq 1$ for all quantum channels to conclude that the nowhere dense property also holds for causal quantum channels in the set of local quantum channels in the pointwise weak and $\sigma$-weak topologies.

\end{proof}

The application of the above theorem relies on the fact that at least one acausal channel exists. In the free real scalar case, this follows rather directly, see \cite{PhysRevD.105.025003} for numerous examples, and corollary \ref{cor:existence_acausal} for application of those results to our case.

From theorem \ref{thrm:causal_meagre} we see that causality is a rare property of a generic local channel. This sharply formalises the intuition one gets by trying to construct causal channels in QFT, where most guesses fail to be causal. Note however, that since quantum channels form a convex set, the set of causal channels is still connected even though it is nowhere dense in the set of all local channels. An example of such a subset is the halfline $\mathbb{R}^+\vec{v}\subseteq \mathbb{R}$ parallel to $\vec{v}$, which is both nowhere dense and connected, as it is a convex set. 

Again, while we have no general proof that there is always an acausal channel, we can state the following in the case of a real scalar field. 
\begin{corollary}\label{cor:existence_acausal}
    Given any globally hyperbolic spacetime $\mathbf{M}$, with a real scalar QFT $\mathfrak{A}$, and any compact subregion $\mathbf{K}$, the set of causal channels is nowhere dense in the set of local channels.
    \begin{proof}
        Let $\mathbf{K}$ be compact. The local algebras for a real scalar field are hyperfinite type $\text{III}_1$ algebras, so to apply theorem \ref{thrm:causal_meagre}, we need only show that there is an acausal unitary in $\operatorname{U}(\mathbf{K})$. By a mild adaptation of claims 10.5 and 10.7 of \cite{albertini2023ideal}, there exists an $f$ supported in $\mathbf{K}$, and $h,g$ spacelike separated and supported in $\mathbf{K}_\text{in/out}$ respectively, such that 
        $\Delta(f,g)\neq 0\neq \Delta(f,h)$. Then, by section III.C \cite{PhysRevD.105.025003}, the unitary $e^{i\phi(f)^2}\in \operatorname{U}(\mathbf{K})$ is acausal, as
        \begin{equation}
            e^{i\lambda \phi(h)} e^{i\phi(f)^2} \phi(g)e^{-i\phi(f)^2}e^{-i\lambda \phi(h)}=\phi(g)-2\Delta(f,g)(\phi(f)+\lambda\Delta(f,h))
        \end{equation}
        depends explicitly on $\lambda$, and so can communicate.
    \end{proof}
\end{corollary}
We see nothing special about the real scalar field, other than its simplicity, and expect similar results to hold for a broad class of QFTs, including interacting ones. 

\section{Conclusions}
We have studied the topology of causal channels in QFT and QM. Starting with a number of structural results relating to the Banach space structure of completely bounded maps, we further specialised to UCP and causal UCP channels. The closure of the $\sigma$-weakly continuous completely bounded maps in CB-norm topologies was used to prove that causal normal channels are nowhere dense in the set of all normal channels. This sharpens Sorkin's impossible operations: the set of local channels that are causal is \textit{nowhere dense}. 

As we have mentioned, we have allowed ourselves two simplifications: normal channels, and von Neumann local algebras. While it may be possible that normality can be relaxed, there are results that suggest that something else must take its place, as without normality, commutation of spacelike separated local channels is not guaranteed \cite{simmons_brukner}. Likewise, $C^*$-algebras are a standard generalisation of von Neumann QFTs, however they lack many of the nice topological properties we have exploited. We will leave the possibility of relaxing these assumptions to future work. 

The preceding results concern abstract normal channels and unitary elements of the local von Neumann algebra. They do not by themselves determine which of these channels arise from specific dynamical models, such as FV-type couplings, perturbative interaction Hamiltonians, or time ordered exponentials of Wick polynomials. It remains open whether common QFT coupling models generate a large, dense, or meagre subset of the (normal) causal channels.

\begin{acknowledgments}
This research was funded in whole or in part by the Austrian Science Fund (FWF) [10.55776/COE1]. For open access purposes, the authors have applied a CC BY public copyright
license to any author accepted manuscript version arising from this submission. 
The author is grateful to Robert Oeckl and  Maria Papageorgiou for helpful feedback and encouragement regarding this manuscript. 
\end{acknowledgments}

\bibliography{references-2, refs2}
\end{document}